\documentclass{sn-jnl}

\usepackage{cite}
\usepackage{amsmath,amssymb,amsfonts}
\usepackage{amsthm}
\usepackage{graphicx}
\usepackage{algorithm}
\usepackage{algorithmic}

\theoremstyle{thmstyleone}%
\newtheorem{theorem}{Theorem}
\newtheorem{lemma}[theorem]{Lemma}
\theoremstyle{thmstyletwo}%
\newtheorem{example}{Example}%
\newtheorem{remark}{Remark}%

\theoremstyle{thmstylethree}%
\newtheorem{definition}{Definition}%

\renewcommand{\vec}{\boldsymbol}

\newcommand{\vc}{\vec{c}}

\newcommand{\ka}[2]{\kappa_{#1}^{#2}}

\begin{document}

\title{Efficient encoding and decoding algorithm for a class of perfect single-deletion-correcting permutation codes}



\author{\begin{center}Minhan Gao and Kenneth W. Shum\footnote{Corresponding author. Email: wkshum [at] cuhk.edu.cn} \\
School of Science and Engineering \\
The Chinese University of Hong Kong (Shenzhen) \end{center}}


\abstract{
A permutation code is a nonlinear code whose codewords are permutation of a set of symbols. We consider the use of permutation code in the deletion channel, and consider the symbol-invariant error model, meaning that the values of the symbols that are not removed are not affected by the deletion.
In 1992, Levenshtein gave a construction of perfect single-deletion-correcting permutation codes that attain the maximum code size. Furthermore, he showed in the same paper that the set of all permutations of a given length can be partitioned into permutation codes so constructed. This construction relies on the binary Varshamov-Tenengolts codes. In this paper we give an independent and more direct proof of Levenshtein's result that does not depend on the Varshamov-Tenengolts code. Using the new approach, we devise efficient encoding and decoding algorithms that correct one deletion.
}

\keywords{Permutation codes, Varshamov-Tenengolts code, rank modulation}

\maketitle

\section{Introduction}

A permutation code over an alphabet of size $n$ is a code in which each codeword contains all elements in the alphabet once and exactly once. The basic properties of permutation codes, with the use of the Hamming distance as a design metric, have been documented in prior works~\cite{Blake74, Blake79}. Moreover, a table of permutation codes with their Hamming distances can be found in~\cite{Smith2012}. 

One key feature of permutation codes is that all symbols are used equally often, and thus can be used to combat jamming. As a result, permutation codes have found use in frequency-hopping digital communication systems, such as power-line communication~\cite{Chu2004}.  More recently, permutation codes have also been applied in wireless communication with short block lengths~\cite{Faure2022, Yang2023}.

Another important application of permutation codes is in the context of Flash memory~\cite{FSM13}. In Flash memory, data is encoded using the relative voltage levels between memory cells, rather than the absolute voltage levels, due to the difficulty in controlling the latter with high accuracy. This naturally leads to the use of permutation codes, as each codeword corresponds to a unique arrangement of the voltage levels. For this application, distance measures such as the Ulam distance and the Kendall's $\tau$ distance can model the system more accurately than the Hamming distance.

In the literature, deletions in permutation codes can be classified into two categories: stable deletion and unstable deletion~\cite{Gabrys16}. In the case of stable deletion, the values of the intact symbols are not affected by the deleted symbol. If the codeword is a permutation, the receiver knows which symbols are missing, but does not know the locations where the deletions occurred. This error model has applications in DNA-based storage, for example. An excellent introduction to single-permutation-correcting permutation codes can be found in~\cite{Sloane02}. In the second, ``unstable'' model, the receiver does not have information about the absolute values of the remaining symbols, but only knows their relative values. Consequently, the channel output is a permutation with a shorter block length. This paper will focus on permutation codes designed for the correction of stable deletions.

In the context of combinatorics, permutation codes can be interpreted as directed $t$-designs and directed $t$-packings~\cite{Mathon1999}. This connection stems from the fact that a permutation code can correct any $d$ deletions if and only if all subsequences of the codewords of length $n-d$ are distinct. A {\em directed $t$-packing} consists of a ground set $\mathcal{X}$ and a collection $\mathcal{B}$ of ordered $k$-subsets of $\mathcal{X}$, where every ordered $t$-subset appears in at most one ordered subset in $\mathcal{B}$. When $|\mathcal{X}|=n=k$, this is equivalent to a permutation code that can correct $n-t$ deletions. Furthermore, if every ordered $t$-subset appears in exactly one ordered $k$-subset in $\mathcal{B}$, then we have a directed $t$-design, or a directed $t$-Steiner system. The corresponding permutation code is also said to be perfect~\cite{Klein2004}.

In~\cite{Levenshtein92}, Levenshtein studied perfect permutation codes that can correct a single deletion~\cite{Levenshtein92}. His construction was based on the binary Varshamov–Tenengolts (VT) code~\cite{VT65}. In this paper, we provide a more direct proof of the correctness of this construction, and devise fast encoding and decoding algorithms that do not rely on the VT code.

We review the basic definitions of the VT code below. A binary VT code of length $n$ is a set of binary vectors $(b_1, b_2, \ldots, b_n)$ that satisfy the equation
$$
\sum_{j=1}^{n} j b_j = k \bmod (n+1)
$$
for some constant $k$. The VT code was originally introduced in~\cite{VT65} to combat asymmetric bit errors. Later, Levenshtein showed in~\cite{Levenshtein66} that the VT code can also correct a single deletion, and provided an efficient decoding method. A fast encoding algorithm for the VT code was later given by Abdel-Ghaffar and Ferreira in~\cite{AF98}.

Building on the VT code, Tenengolts introduced a non-binary version and provided a linear-time decoding algorithm~\cite{Tenengolts84}. Leveraging this non-binary VT code, Levenshtein proved in~\cite{Levenshtein92} that by restricting the codewords to permutations, the result is a perfect 1-deletion permutation code. We will refer to this construction as the Levenshtein's permutation code in this paper.

Given a vector $(x_0, x_1, \ldots, x_{n-1})$, where $x_j$ are integers between 0 and $n-1$, the {\em signature} of this vector is defined as a binary vector of length $n-1$, where the $j$-th component is 1 if $x_j \geq x_{j-1}$, and 0 if $x_j < x_{j-1}$, for $j=1,2,\ldots,n-1$. Levenshtein proved in~\cite{Levenshtein92} that the set of all permutations of length $n$ whose signatures are codewords in a VT code of length $n-1$ forms a permutation code that can correct a single deletion. Furthermore, Levenshtein also showed in the same paper that the set of all permutations of length $n$ can be partitioned into $n$ perfect 1-deletion permutation codes. We will see that this latter result also follows from the new approach proposed in this paper.

While Levenshtein's permutation code construction is well-studied, efficient encoding and decoding algorithms have not been readily available until recently. The work in~\cite{Sun23} provided a linear-time encoding algorithm, but the corresponding decoding algorithm may take $O(n^2)$ time. In this paper, we present both encoding and decoding algorithms that have fast implementations.  Encoding can be done in $O(n^{1+\epsilon})$ time, where $n$ is the length of the permutation, and $\epsilon>0$ is a constant. The time complexity of recovering the data from a deleted codeword is $O(n\log n)$. So, both encoding and decoding have quasi-linear time complexity.

For encoding and decoding algorithms for other nonlinear codes that can correct deletions and insertions, we refer the readers to in~\cite{cai2021correcting, Gabrys2022} for more information.

The extension from permutation codes to multi-permutation codes is a rapidly growing area of research. Motivated by applications in DNA storage, codes that can correct burst deletions are of great practical interest. A codeword in a multi-permutation code contains all the symbols a constant number of times, thus enabling more flexibility in code design in compare to permutation codes. Burst-deletion-correcting permutation and multi-permutation codes have been studied extensively in recent works~\cite{Chee2020, TH2022, Han23, Sun23, Wang24}. Constructions of codes that can correct both deletions and substitutions have also been investigated~\cite{Song2020, smagloy2023single}.

This paper is organized as follows. In Sections~\ref{sec:preliminaries} and~\ref{sec:transposition}, we review the preliminaries for permutations, the symmetric group, and permutation codes. Using a representation of permutations described in Section~\ref{sec:phi}, we then give an alternate description of Levenshtein's permutation code construction in Section~\ref{sec:construction}. The main theorem of this paper is presented in Section~\ref{sec:main}, and the proof of a technical theorem is given in Section~\ref{sec:proof_of_theorem_5}. Finally, in Section~\ref{sec:algorithm}, we describe efficient encoding and decoding algorithms for Levenshtein's permutation code.

\section{Permutation codes with deletion metric}

\label{sec:preliminaries}

A {\em permutation} is a bijective function from a finite set to itself. Given two permutations $\pi$ and $\sigma$ on the finite set $A$, we define the {\em composition} of $\pi$ and $\sigma$ by function composition, and is denoted by $\pi\circ\sigma$, or simply $\pi \sigma$. It is the function defined by $\pi\sigma(x)=\pi(\sigma(x))$ for $x\in A$. 

The set of all permutations on a fixed finite set under composition operation forms a group. The identity permutation $e$ is the identity function, in which every element is mapped to itself. Given any permutation $\pi$ on a finite set, the inverse $\pi^{-1}$ is the permutation such that $\pi \pi^{-1}=e$.

In this paper we consider permutations defined on an $n$-set $\{0,1,2,\ldots, n-1\} \triangleq [n]$. We can represent a permutation $\pi:[n]\rightarrow [n]$ by an array of numbers:
$$
\left( \begin{matrix} 0 \\ 
	\pi(0)\end{matrix} \ 
\begin{matrix} 1 \\ \pi(1)\end{matrix} \ 
\begin{matrix} 2 \\  \pi(2)\end{matrix} \ 
\begin{matrix} \cdots \\ \cdots \end{matrix}
\begin{matrix} n-1 \\  \pi(n-1)\end{matrix}  
\right).
$$
More compactly, we will just write the second row as a vector,
$$
	(\pi(0), \pi(1), \ldots, \pi(n-1)).
$$
This is called the {\em vector representation} of a permutation. The group of all permutations on $[n]$ is called the {\em symmetric group} on $[n]$, and will be denoted by $S_n$. With a slight abuse of notation, we will  use the notation  $(a_0, a_1,\ldots, a_{n-1}) \in S_n$ to mean a bijection that maps $i$ to $a_i$, for $i=0,1,\ldots, n-1$. The parameter $n$ is called the {\em length} of the permutation.

A {\em permutation code} of length $n$ is a subset of permutations in $S_n$. Each permutation in a permutation code is called a {\em codeword}. If we represent the codewords by vector, then a permutation code consists of vectors, and in codeword the symbols $0$, $1$,\ldots, $n-1$ appear once and exactly once. If we consider the Hamming distance, then a permutation code with Hamming distance at least $d$ is a subset $\mathcal{C}$ of $S_n$, such that for any pair of distinct permutations $\pi$ and $\sigma$ in $\mathcal{C}$ agree in at most $n-d$ position, i.e.,
$$
 |\{i\in [n]:\, \pi(i) = \sigma(i)\}| \leq n-d.
$$

In this paper we will consider another metric motivated from coding for the Flash memory. Given a vector of $n$ symbols $(a_0,a_1,\ldots, a_{n-1})\in S_n$, we say that there is a deletion at the $i$-th symbol if what we can observe is the vector
$$
(a_0, a_1,\ldots, a_{i-1}, a_{i+1}, \ldots, a_{n-1}).
$$
The absolute values of the other symbols are not affected, but we do not know the location of the deletion.

If it is known that the original codeword is a permutation, the symbol that is deleted can be uniquely determined. If we know that symbol $s$ is deleted, we can recover the original codeword if the subsequences obtained by removing symbol $s$ from the codewords are all distinct.  We say that a permutation code is {\em single-deletion-correcting} if we can recover the codeword from the deletion of any symbol.

As a result, the maximal cardinality of a single-deletion-correcting permutation code is upper bounded by $(n-1)!$. A single-deletion-correcting permutation code is said to be {\em perfect} if it has cardinality $(n-1)!$.

\begin{example}
	The following 6 codewords in $S_4$ form a perfect single-deletion-correcting permutation codes:
\begin{gather*}
(0,1,2,3), \ (3,0,2,1), \ (3,1,2,0), (2,0,3,1),\ (2,1,3,0),\ (1,0,3,2).
\end{gather*}
We note that if we delete the symbol `0' from each codeword, the resulting sequences are all distinct. Likewise, if we delete symbol `1` (or `2' or `3'), we will obtain six distinct sequences.
\end{example}

\section{Transposition and block  transposition}
\label{sec:transposition}
We introduce several types of important permutations in this section. In the followings, all permutations have length~$n$.

Given two distinct integers between 0 and $n-1$,
we denote the permutation that exchanges the two symbols at locations $i$ and $j$ by $\sigma_{i,j}$, and call it a {\em transposition}. More formally, a transposition $\sigma_{i,j}$ is a bijection defined by $\sigma_{i,j}(i)=j$ and $\sigma_{i,j}(j)=i$ for some $i\neq j$, and $\sigma_{i,j}(x) = x$ for all $x\in[n]\setminus \{i,j\}$. A transposition is called an {\em adjacent transposition} if $j=i+1$. 

A permutation is called a {\em right cyclic shift} if it sends $i$ to $i+1 \bmod n$, for $i=0,1,\ldots, n-1$. In vector notation, a right cyclic shift is presented by
$$
(1,2,\ldots, n-1,0).
$$
We call it a right cyclic shift because it moves the first symbol to the far right. Likewise, a {\em left cyclic shift} is the permutation that maps $i$ to $i-1 \bmod n$ for all~$i$.

A {\em translocation} is a permutation in the form
$$
(0,1,\ldots, i-1, i+1, i+2,\ldots, k-1, k, i, k+1, \ldots, n-1)
$$
were $i$ and $k$ are indices with $i<k$.

Motivated by applications in DNA, we consider another class of permutations that exchange two block of symbols. Let $i$, $j$ and $k$ be three integers between 0 and $n-1$, such that $i<j<k$. We define a {\em block transposition} as a power of a translocation, and it has the form
$$
(0,1,\ldots, i-1, \underbrace{j,j+1,\ldots, k-1},  \underbrace{i,i+1,\ldots, j-1}, k,k+1,\ldots, n-1)
$$
for some $i<j<k$.
When $j=i+1$ and $k=j+1$, a block transposition reduces to a transposition. A block transposition is also called a {\em generalized adjacent transposition}~\cite{Wang2024}.

When $k=n$, we will refer it as {\em suffix block transposition}. For any integer $i=0,1,\ldots, n-1$, and integer $s$ between 1 to $n-i-1$, we use the notation $\kappa_{i,s}$ to stand for the permutation
$$
(0,1,\ldots, i-1, \underbrace{i+s,i+s+1,\ldots, n-1}_{n-i-s},  \underbrace{i,i+1,\ldots, i+s-1}_{s}).
$$
Moreover, we define $\kappa_{i,n-i}$ as the identity permutation. We note that the suffix block transpositions 
$$
\kappa_{i,1},\ \kappa_{i,2}, \ldots, \kappa_{i,n-i}
$$
form a subgroup of $S_n$ with order $n-i$.

\smallskip

Let $\rho$ be an arbitrary permutation in $S_n$, represented in vector notation as
$(\rho_0, \rho_1,\ldots, \rho_{n-1})$. The effect of compose $\rho$ with a transposition $\sigma_{i,j}$ on the right is to exchange the $i$th and $j$th symbols in $\rho$,
$$
\rho \circ \sigma_{i,j} = (\rho_0,\rho_1,\ldots, \rho_{i-1},\rho_j, \rho_{i+1}, \ldots, \rho_{j-1}, \rho_i, \rho_{j+1}, \ldots).
$$
It is the bijection that maps $i$ to $\rho_j$, $j$ to $\rho_i$, and $x$ to $\rho_x$ for $x\in [n]\setminus\{i,j\}$. On the other hand, if we apply transposition $\rho_{i,j}$ to the right of $\rho$, we exchange the two symbols $i$ and $j$ in $\rho$ instead.

Similar remarks go with translocation. In general, composing a translocation on the right is an operation about the locations. Composing a translocation on the right is about the symbol values instead.

\begin{example}
Suppose $\rho$ is the permutation $(2,0,4,1,5,3)$. If we apply transposition $\sigma_{3,4}$ on the right, we obtain
$$
(2,0,4,1,5,3)\circ \sigma_{3,4} = (2,0,4,5,1,3).
$$
If we apply $\sigma_{3,4}$ on the left, the result is
$$
\sigma_{3,4} \circ (2,0,4,1,5,3) = (2,0,3,1,5,4).
$$\end{example}

In this paper, we will interpretation a deletion as a translocation. If the symbol of permutation $\pi$ at position $i$ is deleted, we imagine that the missing symbol is moved to the end of the permutation, and the corresponding permutation is
$$
\pi \circ \kappa_{i,1} = (\pi_{0},\pi_{1},\ldots, \pi_{i-1}, \pi_{i+1}, \pi_{i+2},\ldots, \pi_{n-1}, \pi_{i}).
$$
The subsequence consisting of the first $n-1$ symbols in $\pi \circ \kappa_{i,1}$ are precisely the subsequence obtained by deleting the symbol a position $i$ from~$\pi$.

Given a permutation $\pi$, define the 1-deletion ball centered at $\pi$ with radius 1 as
$$
B(\pi,1) \triangleq \{\pi\circ \kappa_{i,1}:\, i=0,1,2,\ldots, n-1\}.
$$
As a result, a subset $\mathcal{C}$ of permutations in $S_n$ can correct any single deletion if and only if the balls $B(\pi,1)$ for $\pi$ in $\mathcal{C}$ are mutually disjoint. If the balls $B(\pi,1)$, for $\pi\in\mathcal{C}$, form a partition of $S_n$, then $\mathcal{C}$ is a perfect 1-deletion permutation code.

\section{A representation of permutations using suffix block transpositions}
\label{sec:phi}

The construction of perfect single-deletion-correcting permutation codes in this paper stems from a special representation of permutation that is very similar to~\cite{Arnow1990}. 
We will establish a one-to-one correspondence between permutation of length $n$ and a vector of the form $[a_0, a_1,\ldots, a_{n-1}]$, where $a_j$ is an integer between 1 and $j+1$, for each $j=0,1,\ldots, n-1$. Hence, $a_0$ is always equal to 1, $a_1$ may take 1 or 2 as its value, etc. Let $V_n$ denote the set of vectors
\begin{equation}
	V_n \triangleq \{[a_0,a_1,\ldots, a_{n-1}]:\, a_j=1,2,\ldots, j+1, \text{for } j=0,1,\ldots n-1\}.
	\label{eq:V_n}
\end{equation}

We consider the mapping that maps a vector  $\alpha = [a_0,a_1,a_2,\ldots, a_{n-1}]$ in $V_n$  
to a permutation in $S_n$,
\begin{equation}
	[a_0,a_1,a_2,\ldots, a_{n-1}] \mapsto  \kappa_{n-2,a_1} \kappa_{n-3, a_2} \cdots \kappa_{1, a_{n-2}} \kappa_{0, a_{n-1}}.
	\label{eq:inverse_rep}
\end{equation}

We claim that any permutation $\pi = (\pi_0,\pi_1,\ldots, \pi_{n-1})$ in $S_n$ can be written as an image of this map. Let $\pi$ be an arbitrary permutation of length~$n$. We first recall that $\kappa_{0,s}$ is a cyclic shift, for $s=0,1,\ldots, n-1$.  We choose a value of $s$, and call it $a_{n-1}$, such that  $\pi \kappa_{0,a_{n-1}}^{-1}$ has the symbol 0 located at the first position.  Next, we note that a block transposition in the form $\kappa_{1,s}$ fixes the symbol at position 0, and cyclically shift the remaining symbols. We can choose the unique integer $a_{n-2}$ between 1 and $n-1$ such that the second position in $\pi \kappa_{0,a_{n-1}}^{-1} \kappa_{1,a_{n-2}}^{-1}$ is~1. We repeat this process and obtain $a_j$, for $j=2,3,\ldots, n-1$, so that
$$
\pi \kappa_{0,a_{n-1}}^{-1} \kappa_{1,a_{n-2}}^{-1} \kappa_{2,a_{n-3}}^{-1} \cdots \kappa_{n-2,a_1}^{-1}
$$
is the identity permutation.
 $(0,1,2,\ldots, n-1)$. Hence, every permutation $\pi$ can be expressed as
$$
\pi = \kappa_{n-2,a_1} \kappa_{n-3, a_2} \cdots \kappa_{1, a_{n-2}} \kappa_{0, a_{n-1}}
$$
for some $[a_0,a_1,\ldots, a_{n-1}]$ in $V_n$.

Since the mapping in \eqref{eq:inverse_rep} is surjective, and $V_n$ and $S_n$ has the same cardinality, this map is indeed a bijection between $V_n$ and $S_n$. As a result, every permutation can be represented uniquely by a numerical vector in $V_n$. We record this fact in the following theorem.

\begin{theorem}
	Let $V_n$ denote the set of vectors as in \eqref{eq:V_n}. The map in \eqref{eq:inverse_rep} is a bijection from the set $V_n$ defined in \eqref{eq:V_n} to the symmetric group $S_n$.
\end{theorem}

\begin{definition}
	For any $\pi$ in $S_n$, we denote the unique vector $[a_0,a_1,\ldots, a_{n-1}]$ in $V_n$ that corresponds to $\pi$  by $R(\pi)$.
\end{definition}

The representation of the 24 permutations in $S_4$ is given in  Table~\ref{table:S4}.

\begin{table}
	\caption{Representation of $S_4$ using the bijection in \eqref{eq:inverse_rep}}
	\label{table:S4}
	\begin{tabular}{|c|c|| c|c|| c|c|| c|c| } \hline
		$\pi$ & $R(\pi)$ & $\pi$ & $R(\pi)$ & $\pi$ & $R(\pi)$ & $\pi$ & $R(\pi)$\\ \hline
		$(3,2,1,0)$ & [1,1,1,1] & (2,1,0,3) & [1,1,1,2] & (1,0,3,2) & [1,1,1,3]  & (0,3,2,1) & [1,1,1,4]\\
		$(2,3,1,0)$ & [1,2,1,1] & (3,1,0,2) & [1,2,1,2]& (1,0,2,3) & [1,2,1,3]& (0,2,3,1) & [1,2,1,4]\\
		$(2,1,3,0)$ & [1,1,2,1] & (1,3,0,2) & [1,1,2,2]& (3,0,2,1) & [1,1,2,3]& (0,2,1,3) & [1,1,2,4]\\
		$(3,1,2,0)$ & [1,2,2,1] & (1,2,0,3) & [1,2,2,2]& (2,0,3,1) & [1,2,2,3]& (0,3,1,2) & [1,2,2,4]\\
		$(1,3,2,0)$ & [1,1,3,1] & (3,2,0,1) & [1,1,3,2]& (2,0,1,3) & [1,1,3,3]& (0,1,3,2) & [1,1,3,4]\\
		$(1,2,3,0)$ & [1,2,3,1] & (2,3,0,1) & [1,2,3,2]& (3,0,1,2) & [1,2,3,3]& (0,1,2,3) & [1,2,3,4] \\ \hline
	\end{tabular}
\end{table}

Given a vector $[a_0,a_1,\ldots, a_{n-1}]$ in $V_n$, we can generate the corresponding permutation in $n-1$ steps. In the $j$-th step, where $j=1,2,\ldots, n-1$, we apply the action of $\kappa_{n-1-j,a_{j}}$ from the right.

\begin{example} \label{ex:alpha_to_permutation}
	We illustrate how to obtain the permutation associated to a vector in $V_n$. The vector $[a_0,a_1,a_2,a_3,a_4]=[1,1,1,3,2]$ represents a permutation in $S_5$. To obtain this permutation, we start from the identity permutation $(0,1,2,3,4)$. 
	
	Step 1: As $a_1=1$, we exchange 3 and 4 to obtain
	$$(0,1,2,4,3).$$
	
	Step 2:  Because $a_2=1$, we apply a block transposition to the last three symbols to get
	$$
	(0,1,4,3,2).
	$$
	
	Step 3: Apply $\kappa_{2,3}$. The permutation after this step is
	$$
	(0,2,1,4,3)
	$$
	
	Step 4: Finally, we cyclically shift the permutation until the symbols is located at position 3. The final result is
	$$
	(1,4,3,0,2).
	$$
\end{example}

In the previous example we can see that the relative positions of symbols $2$, $3$, and $4$ are affected by the value of $a_1$ only. If we cyclically shift the permutations $(0,1,4,3,2)$, $(0,2,1,4,3)$ and $(1,4,3,0,2)$ in steps 2, 3 and 4. such that the the symbol 2 is on the first position and read out the symbols from left to right, the order of getting the three symbols 2, 3, 4 is 2, 4 and 3.

In general, each component in $R(\pi)=[a_0,a_1,\ldots, a_{n-1}]$ can be determined from the vector representation of $\pi$ using the following procedure. Suppose we want to compute $a_j$, where $j$ ranges from 0 to $n-2$. We cyclically rotate permutation $\pi$ until symbol $n-1-j$ is in the first position on the left. Then $a_j$ is the number of symbol to the left of symbol $n-2-j$ that is strictly larger than $n-2-j$. For $j=n-1$, we append an extra symbol `$e$' at the end of $\pi$, and cyclically shift the resulting sequence until the first position is the symbol $0$. We set $a_{n-1}$ to be the number of symbols that are to the left of the symbol `$e$'.

We note that the smallest value of $a_j$ is least 1, because the first symbol $n-1-j$ is larger than $n-2-j$, and the largest value is $j+1$.

\begin{example} \label{ex:e}
	Consider the permutation $\rho = (3,4,0,1,2)$. 
	We compute the components $a_j$ in  $R(\rho)=[a_0,a_1,\ldots,a_4]$ below.
	
	$j=0$. Cyclically shift $\rho$ to $(4,0,1,2,3)$. The only symbol to the left of 3 that is larger than 3 is the symbol 4. Hence, $a_0=1$.
	
	$j=1$. Because the permutation $\rho$ starts with the symbol 3, we do not need to perform any cyclic shift. The value of $a_1$ is equal to 2, as both symbols 3 and 4 are located on the left side of symbol~2.
	
	$j=2$. Cyclically shift $\rho$ to  $(2,3,4,0, 1)$. We see that symbols 2, 3 and 4 are to the left of 1. Consequently, we have $a_2=3$.
	
	$j=3$. Cyclically rotate $\rho$ to $(1,2,3,4,0)$. Since the symbol 0 is at the end of the sequence, there are four symbols to the left of 0 that are larger than 0, and we have $a_3=4$. 
	
	$j=4$. After appending an extra symbol $e$ at the end, we cyclically shift it such that the first symbol is 0, i.e.,  $(0,1,2,e,3,4)$. Symbols 0, 1 and 2 are located to the left of~$e$. Therefore $a_4 = 3$.
	
	In conclusion, the representation $R(\rho)$ is $[1,2,3,4,3]$.
\end{example}

Because the positions of a symbol can be obtained from the inverse permutation $\pi^{-1}$, we can formulate the procedure in the previous paragraphs in terms of~$\pi^{-1}$. In the followings, we will use the notation $[x\bmod n]$ to  stands for the unique integer between 0 and $n-1$ such that $x - [x\bmod n]$ is divisible by~$n$.

\begin{theorem}
	Let $\pi$ be a permutation $\pi$ in $S_n$, and $(c_0,c_1,\ldots, c_{n-1})$ be the inverse of~$\pi$. Define $c_{-1}=0$.	
	The components $a_i$ in the vector $R(\pi)$, for $i=0,1,\ldots, n-1$ is equal to the number of indices $k\in\{0,1,2,\ldots, i\}$ such that
	\begin{equation}
		[c_{n-1-k} - c_{n-1-i} \bmod n] \leq [c_{n-2-i} - c_{n-1-i} \bmod n],
		\label{eq:alpha}
	\end{equation}
	and moreover, we have $a_{n-1} = n - c_{0}$.
	\label{thm:R_using_inverse}
\end{theorem}

\begin{proof}
By the definition of inverse permutation, the value of $c_i$ is the position of symbol $i$ in $\pi$. The left-hand side of \eqref{eq:alpha} represents the relative distance from symbol $n-1-i$ to symbol $n-2-i$. The condition in \eqref{eq:alpha} holds if the symbol $n-1-k$ is within the range from $n-1-i$ to $n-2-i$, and we count the number of symbols in $\{n-1, n-2, \ldots, n-1-i\}$ that satisfy this condition.

We can see that $a_{n-1} = n-c_0$ from the definition of $a_{n-1}$, because $c_0$ represents the location of the symbol 0 in $\pi$. We can check that when $i=n-1$, all indices $k$ in $\{0,1,\ldots, n-1\}$ satisfy the condition in~\eqref{eq:alpha}. Hence, the calculation in Theorem~\ref{thm:R_using_inverse} is consistent with the definition of $a_{n-1}$ in~$R(\pi)$.
\end{proof}

\begin{remark}
We remark that there is at least one $k$ in the range between 0 and $i$ that satisfies the condition in \eqref{eq:alpha}. When $k=i$, the left-hand side is equal to 0, and hence the condition in \eqref{eq:alpha} must hold. This is consistent with the fact that $a_i$ is larger than equal to~1.
\end{remark}

The condition in Theorem~\ref{thm:R_using_inverse} can be written equivalently in several equivalent forms. For example, $a_i$ is also equal to the number of indices $k$ from 0 to $i$ such that
\begin{equation}
[c_{n-1-k} - c_{n-2-i} \bmod n] \leq [c_{n-1-i} - c_{n-2-i} \bmod n].
\label{eq:alpha2}
\end{equation}

\section{An alternate description of Levenshtein's permutation codes}
\label{sec:construction}
Using the representation $R$ defined in the previous section, we give an alternate description of Levenshtein's construction of permutation code.

\smallskip

\noindent {\bf Construction.} Let $n$ be a positive integer, and  $T$ be any integer.  We define $\mathcal{W}_{T,n}$ as the subset of $V_n$ consisting of vectors $[a_0, a_1, a_2,\ldots, a_{n-1}]$ that satisfy
$$
a_0+a_1+a_2+\cdots + a_{n-1} \equiv T \bmod n.
$$
We construct a permutation code $\mathcal{C}_{T,n}$ by taking the inverse of the permutations corresponding to the vector in $\mathcal{W}_{T,n}$,
\begin{equation}
\mathcal{C}_{T,n} \triangleq \{(R^{-1}(\alpha))^{-1}: \alpha \in \mathcal{W}_{T,n}\}.
\label{eq:construction}
\end{equation}

The inverse  function $R^{-1}$ is well-defined, as $R$ is a bijective function.

\noindent {\it Remark:} The value of $a_0$ is  constantly equal to 1. Even though it contains no information, we will see later that the inclusion is $a_0$ is crucial, and it can simplify the proof significantly. Moreover, since $T$ is an arbitrary constant, whether we add $a_0$ on the left-hand side has no effect on the theory we are going to develop below.

\begin{example}
Consider code length $n=4$. By the above construction, we can construction four 1-deletion permutations code, corresponding to $T=0$, $T=1$, $T=2$, and $T=3$. We list the codewords in Table~\ref{table:examples_of_length_4}.
\end{example}

\begin{table}
\caption{Four perfect 1-deletion permutation codes.}
\label{table:examples_of_length_4}
	\begin{tabular}{|c|c || c|c || c|c || c|c|} \hline
$\pi$ in $\mathcal{C}_{0,4}$ & $R(\pi)$  & 	$\pi$ in $\mathcal{C}_{1,4}$ & $R(\pi)$ & $\pi$ in $\mathcal{C}_{2,4}$ & $R(\pi)$ & 	$\pi$ in $\mathcal{C}_{3,4}$ & $R(\pi)$\\ \hline
(3,2,1,0) & [1,1,1,1] & (2,1,0,3) & [1,1,1,2] & (1,0,3,2) & [1,1,1,3] & (0,3,2,1) & [1,1,1,4] \\
(0,2,1,3) & [1,1,2,4] & (3,1,0,2) & [1,1,2,1] & (2,0,3,1) & [1,1,2,2] & (1,3,2,0) & [1,1,2,3]\\
(1,2,0,3) & [1,1,3,3] & (0,1,3,2) & [1,1,3,4] & (3,0,2,1) & [1,1,3,1] & (2,3,1,0) & [1,1,3,2]\\
(0,3,1,2) & [1,2,1,4] & (3,2,0,1) & [1,2,1,1] & (2,1,3,0) & [1,2,1,2] & (1,0,2,3) & [1,2,1,3]\\
(1,3,0,2) & [1,2,2,3] & (0,2,3,1) & [1,2,2,4] & (3,1,2,0) & [1,2,2,1] & (2,0,1,3) & [1,2,2,2]\\
(2,3,0,1) & [1,2,3,2] & (1,2,3,0) & [1,2,3,3] & (0,1,2,3) & [1,2,3,4] & (3,0,1,2) & [1,2,3,1]\\ \hline 
	\end{tabular}
\end{table}

A permutation code obtained by this construction  is indeed a Levenshtein's permutation code. We illustrate the proof idea through an example.

We view a permutation as a periodic sequence, by extending the domain $\{0,1,\ldots, n-1\}$ of a permutation to $\mathbb{Z}$ in the natural way, and scan from right to left in the periodic sequence. In the followings, let $\rho^{-1}$ be a given permutation, which is a codeword in a permutation code, and represent $\rho$ as a a periodic sequence. An example for $\rho = (1,3,4,0,5,2)$ is shown in Fig.~\ref{fig:reverse}. The codeword is $\rho^{-1} = (3,0,5,1,2,4)$.

To facilitate the description, we define $\tilde{a}_{j}$ as $a_{n-1-j}$,
for $j=0,1,\ldots, n-1$. In order to make the connection to VT code, we also need to define the signature. Let $(c_0,c_1,\ldots, c_{n-1})$ be the vector representation of $\rho^{-1}$. For $j=1,\ldots, n-1$, we let
$$
\tilde{b}_j \triangleq \begin{cases}
	1 & \text{if } c_j > c_{j-1} \\
	0 & \text{if } c_j < c_{j-1}.
\end{cases}
$$


\begin{figure}
	\begin{center}
		\includegraphics[width=10cm]{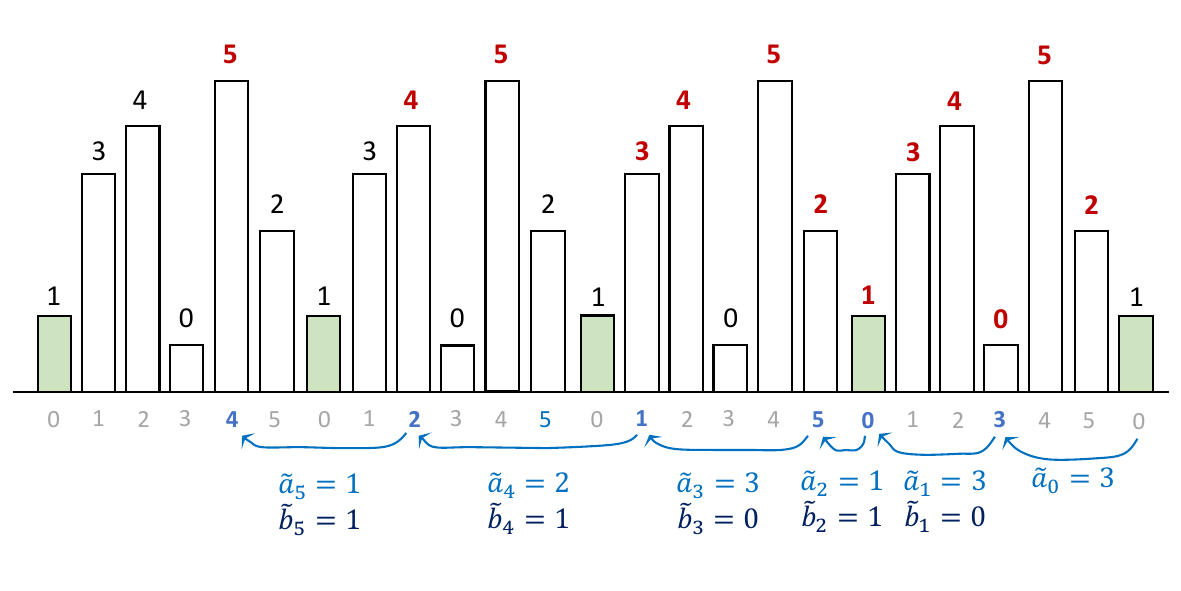}
	\end{center}
	\caption{Reverse calculations codeword is $\rho^{-1} = (3,0,5,1,2,4)$, whose inverse is $\rho = (1,3,4,0,5,2)$. }
	\label{fig:reverse}
\end{figure}


As an example, the calculations of $\tilde{a}_j$ for $\rho= (3,0,5,1,2,4)$ is shown below.

\begin{align*}
	\tilde{a}_0 = 3,  &\qquad  (\mathbf{0},\mathbf{5}, \mathbf{2} ,e,1,3,4) \\
	\tilde{a}_1 = 3, & \qquad  (\mathbf{1},\mathbf{3},\mathbf{4},0,5,2) \\
	\tilde{a}_2 = 1, & \qquad  (\mathbf{2},1,3,4,0,5) \\
	\tilde{a}_3 = 3, & \qquad  (\mathbf{3},\mathbf{4},0,\mathbf{5},2,1) \\
	\tilde{a}_4 = 2, & \qquad  (\mathbf{4},0,\mathbf{5},2,1,3) \\
	\tilde{a}_5 = 1, & \qquad  (\mathbf{5},2,1,3,4,0) 
\end{align*}

The total number of symbols that appear in the procedure in the previous paragraph is exactly equal to the parity of the permutation.  We observe that the symbols in boldface in the first two permutations for $\tilde{a}_0$ and $\tilde{a}_1$ are 0, 1, 2, 3, 4, and 5. The symbols that appear in the third and fourth permutations for $\tilde{a}_2$ and $\tilde{a}_3$ are 2, 3, 4, and 5. In the second last permutations, the symbols 4 and 5  are highlighted, and in the last permutation, only the symbol 5 is counted.


We can arrange th symbols that are counted in a $6\times 4$ table.
$$
\begin{array}{|c|c|c|c|} \hline
	5& 5 & 5 & 5 \\ \hline
	4& 4 & 4 & * \\ \hline
	3& 3 & * & * \\ \hline
	2& 2 & * & * \\ \hline
	1& * & * & * \\ \hline
	0& * & * & * \\ \hline 
\end{array}
$$
Note that we start a new column whenever we have $c_j > c_{j-1}$, i.e., when $b_j = 1$.

If we filling the rest of the table by asterisks, then the number of asterisks is equal to
$$
\sum_{j=1}^{n-1} j \tilde{b}_j.
$$
Since the total number of asterisks and the integers in the table is a multiple of $n$, we have
$$
\sum_{j=0}^{n-1} \tilde{a}_j + \sum_{j=1}^{n-1} j \tilde{b}_j \equiv 0 \bmod n.
$$

Finally, we note that the value of $\tilde{a}_j$ is equal to $a_{n-1-j}$. Hence, the sum of $\tilde{a}_j$ over all $j$ is the same as the sum of $a_j$ over all $j$,
$$
\sum_{j=0}^{n-1} a_j = \sum_{j=0}^{n-1} \tilde{a}_j,
$$
and is equal to the parity of the permutation. If we construct a permutation code using the construction in Section~\ref{sec:construction}, the parity $\sum_{j=0}^{n-1} a_j$ is a constant $T$.
Consequently, the sum  $\sum_{j=1}^{n-1} j \tilde{b}_j$ is equal to a constant $-T \bmod n$. This proves that the permutation codes constructed in this section are Levenshstein codes.

\section{Perfect 1-deletion permutation code}
\label{sec:main}

The next theorem is the main theorem in this paper. This gives a more direct proof that Levenshtein's permutation code is a perfect 1-deletion code. The original proof in~\cite{Levenshtein92} involves the binary VT codes, as Levenstein's construction depends on the binary VT codes. In contrast, the new approach depends on the properties of the representation function $R(\rho)$. Another advantage of the new proof is that it gives a decoding algorithm immediately.

\begin{theorem} The permutation code $\mathcal{C}_{T,n}$ satisfies the followings properties.
	
	\begin{enumerate}
		\item If $T_1 \not\equiv T_2 \bmod n$, the two codebooks $\mathcal{C}_{T_1,n}$ and $\mathcal{C}_{T_2,n}$ are disjoint.	\item 	For fixed $T$, the permutation code $\mathcal{C}_{T,n}$ contains $(n-1)!$ codewords, and is able to correct a single deletion.
	\end{enumerate}
\end{theorem}

It is easy to see the first property, because, if $T_1$ and $T_2$ are not congruent modulo $n$, the two sets $\mathcal{W}_{T_1,n}$ and $\mathcal{W}_{T_2,n}$ are mutually disjoint. Since the mapping $R$ is a bijection, the two preimages $R^{-1}(\mathcal{W}_{T_1,n})$ and 
$R^{-1}(\mathcal{W}_{T_2,n})$ are mutually disjoint. The disjointness is maintained after taking inverse. 

To show $|\mathcal{C}_{T,n}| = (n-1)!$, it is equivalent to show that $\mathcal{W}_{T,n}$ has cardinality $(n-1)!$. In the followings, we use the notation $\langle x \rangle_n$ to represent the remainder of $x$ after division by $n$, when $x$ is not divisible by $n$, and we define $\langle x \rangle_n\triangleq n$ if $x$ is divisible by $n$.
For fixed $T$, we can exhaust all vectors in $\mathcal{W}_{T,n}$ by, for $j=1,2,\ldots,n-2$, setting $a_j$ to be any integer between 1 and $j+1$, and noting that the vector
$$
[1,a_1,a_2,\ldots, a_{n-2}, \langle T-1-a_1-a_2-\ldots- a_{n-2} \rangle_n]
$$
is in $\mathcal{W}_{T,n}$. This proves $\mathcal{W}_{T,n}$ contains $2\times 3 \times \cdots \times (n-1) = (n-1)!$ distinct vectors.

From the remark at the end of Section~\ref{sec:transposition}, we see that we can represent a deletion operation by a block transposition in the form $\kappa_{d,1}$, where $d \in \{0,1,\ldots, n-1\}$ is the index of the position that is deleted. Suppose $\rho^{-1}$ is the transmitted codeword, i.e., $R(\rho)$ is in $\mathcal{W}_{T,n}$. Furthermore, suppose the symbol at position $d$ is deleted. the last entry in the vector representation of $\rho^{-1} \kappa_{d,1}$ is the deleted symbol, while the first $n-1$ entries form the observed sequence after deletion. Hence, the task of decoding is to recover the original codeword $\rho^{-1}$ from $\rho^{-1} \kappa_{d,1}$. The value of $d$ is unknown to the decoder.

The first step in the decoding process is to take the inverse of $\rho^{-1} \kappa_{d,1}$, 
$$
(\rho^{-1} \kappa_{d,1})^{-1}  = \kappa_{d,1}^{-1} \rho.
$$
The vector representation of  $\kappa_{d,1}^{-1}$ is
$$
(0,1,\ldots, d-1, n-1, d, d+1,\ldots, n-2)
$$
The effect of multiplying $\kappa_{d,1}^{-1}$ from the left of $\rho$ is to replace the symbol $d$ in the vector representation of $\rho$ by $n-1$, and replace all symbols $s$ whose value is larger than $d$ by $s-1$. We illustrate the idea in Fig.~\ref{fig:diagram}. The symbol `4' at position 2 in $\rho^{-1}$ is deleted and moved to the end. If we take the inverse of the permutation, this action amounts to pushing the value at position 4 in $\rho$ to the maximum.

\begin{figure}
	\begin{center}
		\includegraphics[width=9cm]{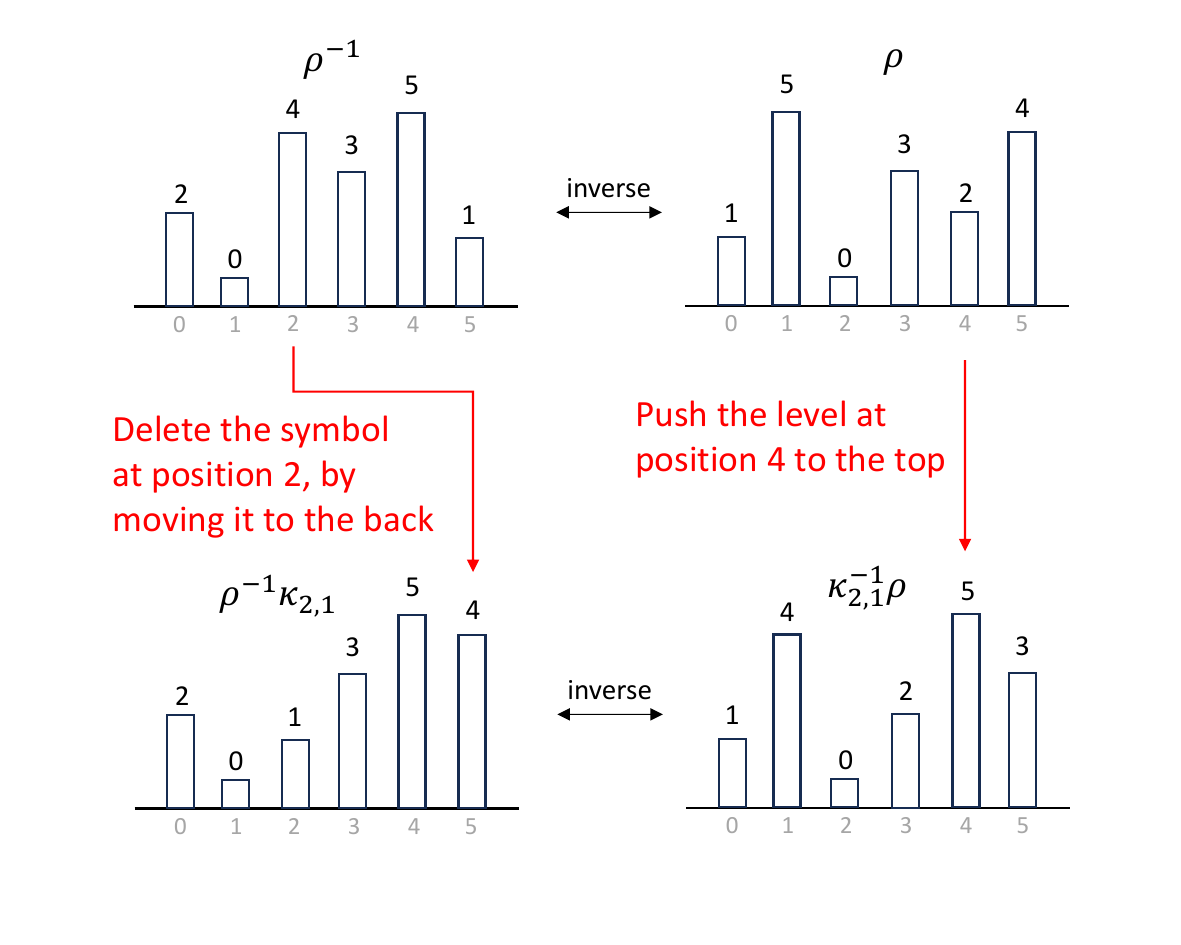}
	\end{center}
	\caption{The effect of a symbol deletion.}
	\label{fig:diagram}
\end{figure}

The second observation is that the block translocation $\kappa_{d,1}^{-1}$ can be decompose as the product of adjacent transpositions
\begin{equation}
\kappa_{d,1}^{-1} = \sigma_{d,d+1} \circ \sigma_{d+1,d+2} \circ \cdots\circ\sigma_{n-3,n-2} \circ \sigma_{n-2,n-1}
\label{eq:transposition_product} 
\end{equation}
In other words, the action of $\kappa_{d,1}^{-1}$ (as a multiplication from the left) on a permutation $\rho$ is to first exchange the symbols $n-1$ and $n-2$, then exchange symbols $n-3$ and $n-2$, and repeat this process until we exchange symbols $d$ and $d+1$ at the end (See Fig.~\ref{fig:bubble}).

\begin{figure}
\begin{center}
	\includegraphics[width=6cm]{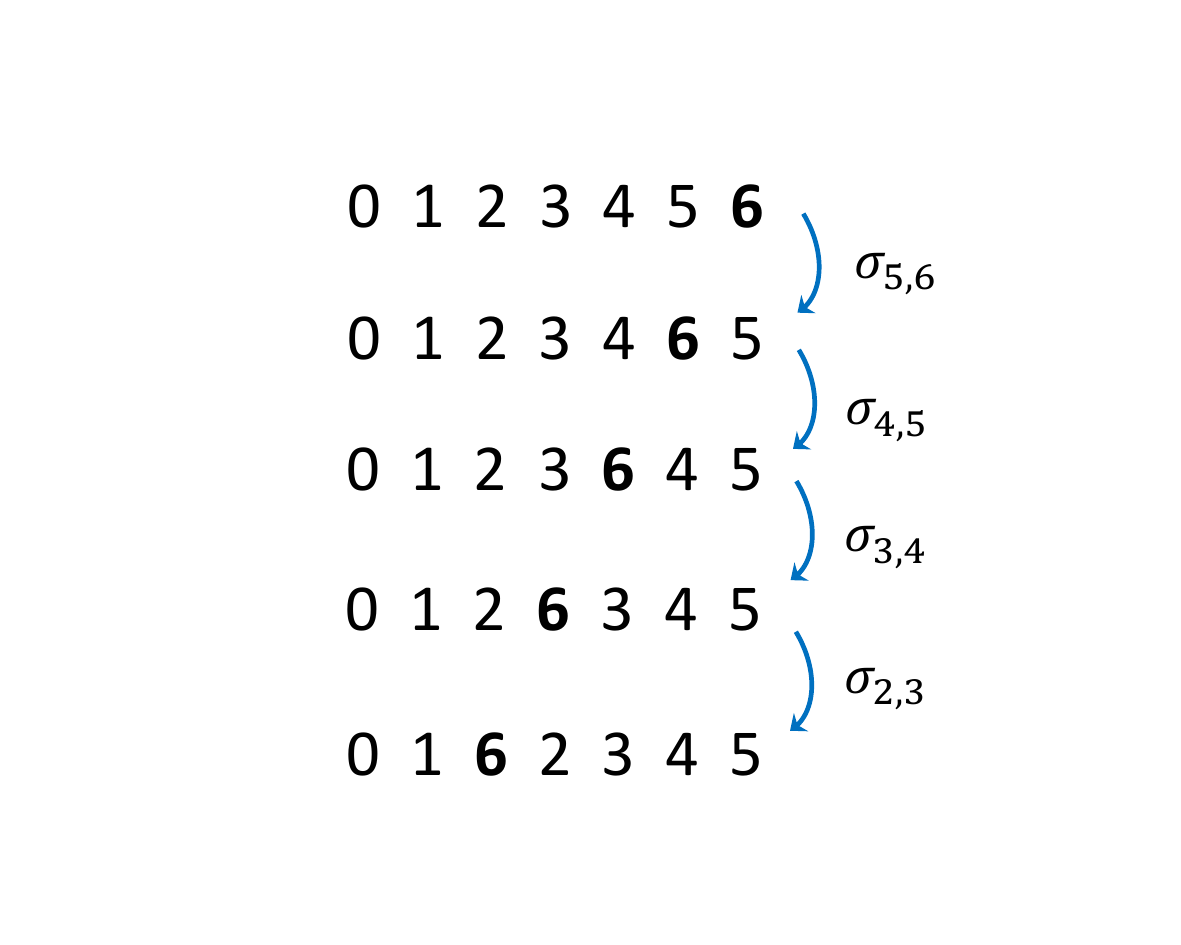}
\end{center}
\caption{Illustration of Equation \eqref{eq:transposition_product}}
\label{fig:bubble}
\end{figure}

For notation convenience, we define the {\em parity} of a
vector $\alpha = [a_0,a_1,\ldots, a_{n-1}]$ in $V_n$ by
$$
\textit{parity}(\alpha) \triangleq a_0+a_1+a_2+\cdots + a_{n-1},
$$
which is the sum of the components in the vector $\alpha$.
We also define the parity of a permutation $\pi$ by $\textit{parity}(R(\pi))$.
Hence, the permutation code $\mathcal{C}_{T,n}$ is the set of permutations whose parity mod $n$ is a constant $T$.

The general idea of the decoding algorithm is as follows. If there is no deletion, we can recover the data from computing the inverse of the permutation and then apply the function $R$. Otherwise, if there is a deletion, we can determine the permutation $\kappa_{d,1}^{-1} \rho$ (as shown  on the lower right corner of Fig.~\ref{fig:diagram}). The location of the largest value corresponds to the value of the symbol that is deleted. We graduate decrease the level at this position to the ground level. We will obtain $n$ permutations, and one of them has the correct parity. This procedure will give the correct answer if the parity of these $n$ permutations modulo $n$ are mutually distinct.

\begin{theorem}
Let $\rho$ be a given permutation. 
Define $n$ permutations $\pi^{(0)}$, $\pi^{(1)},\ldots,\pi^{(n-1)}$  by
$$
\pi^{(j)}  \triangleq  \kappa_{n-j-1,1} \circ \rho,
$$
for $j=0,1,\ldots, n-1$.
Then the parity of $\pi^{(0)}$, $\pi^{(1)},\ldots, \pi^{(n-1)}$ are distinct. Moreover, the set
$$
\{ \textit{parity}(\pi^{(j)}) : j=0,1,\ldots, n-1\}
$$
is a set consisting of $n$ consecutive integers.
\label{thm:distinct_parity}
\end{theorem}

We note that in Theorem~\ref{thm:distinct_parity}, the permutation $\pi^{(0)}$ is the same as $\rho$, because $\kappa_{n-1,1}$ is the identity permutation.

We will describe a procedure that computes the numbers $\textit{parity}(\pi^{(j)})$, for $j=0,1,2,\ldots, n-1$, and then see that these numbers, after some sorting, form a sequence of consecutive integers.

Using the fact that the block transposition $\kappa_{d,1}$ can be decomposed into product of adjacent transpositions as \eqref{eq:transposition_product}, we can obtain $\pi^{(j+1)}$ from $\pi^{(j)}$ by multiplying the adjacent transposition from the left. We can alternately define the permutations $\pi^{(j)}$'s recursively by
$$
\begin{cases}
	\pi^{(0)}  \triangleq \rho  & \\	
	\pi^{(j)}  \triangleq  \sigma_{n-j-1,n-j} \circ \pi^{(j)} &
	\text{for } j=1,2,\ldots, n-1.
\end{cases}
$$
We will view the sequence of permutations
$$
\pi^{(0)}, \ \pi^{(1)}, \pi^{(2)}, \ldots , \pi^{(n-1)}
$$
as the states of a process obtained by decreasing the largest value in the vector representation of $\rho$, and consider change of the corresponding $\textit{parity}(R(\pi^{(j)}))$ as $j$ increases from 0 to $n-1$.


\smallskip

It will be more convenient if we extend the permutation as a periodic sequence. An example of $\rho = (0,4,1,3,2)$ is shown in Fig.~\ref{fig:sign}. The representation of $\rho$ using the mapping $R$ is $[a_0,a_1,\ldots, a_4] = (1,1,2,3,5)$, is shown in red in Fig.~\ref{fig:sign}. We can give an alternate way to compute the representation $R(\rho)$. We scan through the diagram that displays $\rho$ as a periodic sequence from left to right. For each pair of consecutive integers $(i,i-1)$, for $i\geq 1$, we count the number of symbols that are larger than or equal to $i-1$ as we go from symbol $i$ to symbol $i-1$ in the diagram (the starting symbol $i$ is not counted). This gives the value of $a_{n-1-i}$. The value of $a_{n-1}$ is the number of locations from the symbol $0$ to the beginning of the next period.

\begin{figure}
	\begin{center}
		\includegraphics[width=10cm]{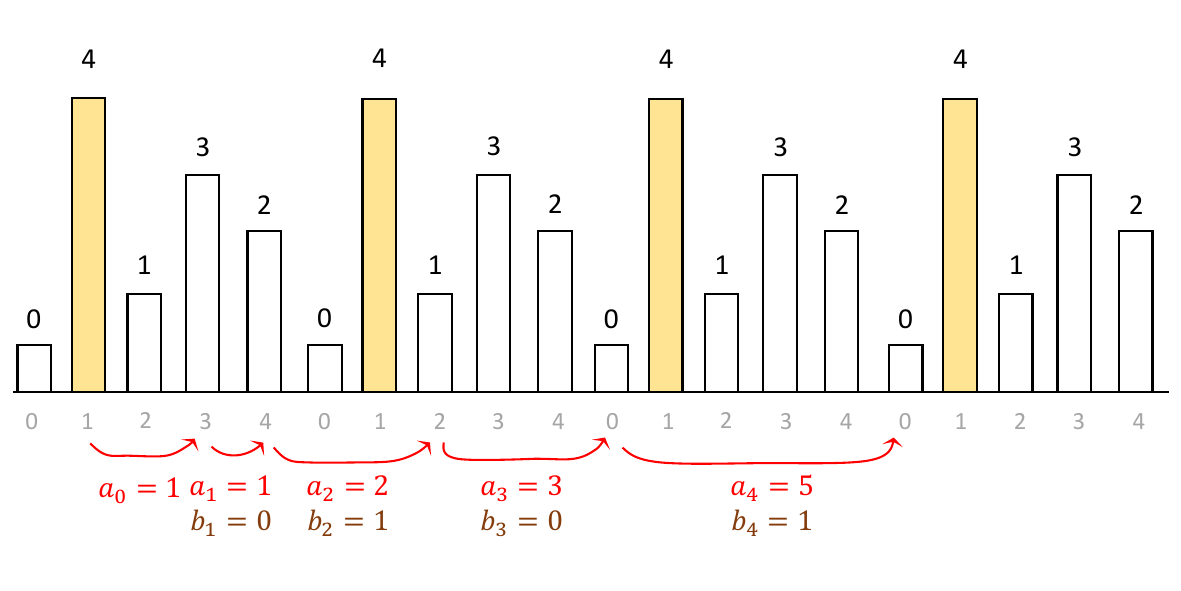}
	\end{center}
	\caption{Periodic presentation of the permutation $\rho = (0,4,1,3,2)$. The locations with the largest value is highlighted.}
	\label{fig:sign}
\end{figure}


We define an important sequence of numbers $b_1$, $b_2,\ldots, b_{n-1}$ for a given permutation $\rho$. For $j=1,2,\ldots, n-1$, we let $b_j$ to be equal to 0 if, when we move from symbol $n-j$ to the nearest symbol $n-j-1$ to the right, we do not pass through the largest symbol $n-1$. Otherwise, we set $b_j =1$.

As the inverse permutation $\rho^{-1}$ encodes the positions of the symbol in the vector representation of $\rho$, we can  give the definition in terms of $\rho^{-1}$. We define, for $i=0,1,2,\ldots, n-2$,
\begin{equation}
	b_{n-1-i} \triangleq \begin{cases}
		0 & \langle [\rho^{-1}(i)-\rho^{-1}(n-1) \bmod n]  + [ \rho^{-1}(i-1)- \rho^{-1}(i) \bmod n] <n , \\ 
		1 & \text{otherwise},
	\end{cases}
	\label{eq:b}
\end{equation}
which is equivalent to
\begin{equation}
b_{n-1-i} \triangleq \begin{cases}
0 & \langle [\rho^{-1}(n-1)-\rho^{-1}(i) \bmod n]  >  [ \rho^{-1}(i-1)- \rho^{-1}(i) \bmod n]  .\\ 
1 & \text{otherwise}.
\end{cases}
\label{eq:c}
\end{equation}
By convention, we define $\rho^{-1}(-1)$ as $n$.



As an example, the permutation $(0,4,1,3,2)$ in Fig.~\ref{fig:sign}, the sequence of integers $b_1, \ldots, b_4$ is $b_1=0$, $b_2=1$, $b_3=0$, $b_4=1$.

The next theorem is the key to the decoding algorithm.

\begin{theorem} \label{thm:distint_values}
For $i=1,2,\ldots, n-1$, we have
\begin{equation}
\textit{parity}(R(\pi^{(i)})) =
 \textit{parity}(R(\pi^{(0)})) + i(1-b_i)- \sum_{j=1}^i b_j  .
 \label{eq:distinct_values}
\end{equation}
\end{theorem}

We note that the sum $ \sum_{j=1}^i b_j$ represents the number of indices $j\in \{1,2,\ldots, i\}$ such that $b_j=1$.  If $b_i=1$, the calculations in \eqref{eq:distinct_values} amounts to subtracting the sum  $ \sum_{j=1}^i b_j$ from  $\textit{parity}(R(\pi^{(0)}))$.
Likewise, when $b_i=0$ the value of $i-\sum_{j=1}^ib_j$ is the number of indices $j \in \{1,2,\ldots, i\}$ such that $b_j=0$. In this case, we add the number of $b_j$'s, for $j=1,2,\ldots, i$, that are equal to 0, to $\textit{parity}(R(\pi^{(0)}))$.

The proof of Theorem~\ref{thm:distint_values} will be given in the next section.

%

Given a sequence of symbols of length $n-1$, obtained from deleting a symbol from a codeword of $\mathcal{C}_{T,n}$. Since there is only one deletion, we can uniquely determine the deleted symbol. We insert the deleted symbol back to the codeword and consider all $d$ possible ways to put it back. For each case, we compute the inverse permutation and the associated representations using the mapping $R$. The resulting representations are precisely $\pi^{(1)},\ldots \pi^{(n-1)}$, for $j=1,2,\ldots, n-1$, as defined in Theorem~\ref{thm:distinct_parity}. Since exactly one of the permutation $\pi^{(j)}$'s has parity congruent to $T$ mod~$n$, we return the inverse of this permutation as the decoded codeword. 

By the construction of $\mathcal{C}_{T,n}$ and the property proved in Theorem~\ref{thm:distinct_parity}, the decoding procedure described in the previous paragraph will always give the correct answer. Because a decoding algorithm that correct a single deletion exists for permutation code $\mathcal{C}_{T,n}$, this completes the proof of Theorem~\ref{thm:distinct_parity}.

\section{Proof of Theorem~\ref{thm:distint_values}}
\label{sec:proof_of_theorem_5}
	
Let $\rho = (\rho_0,\ldots, \rho_{n-1})$ be a fixed permutation in $S_n$. We will write the representation $R(\rho)$ as
$$
R(\rho) = \alpha = [a_0, a_1,\ldots, a_{n-1}].$$
Let $b_1,\ldots, b_{n-1}$ be the zero-one values computed by~\eqref{eq:b}.

We want to calculate the parity of $\pi^{(i)}  = \kappa_{n-i-1,1} \rho$, for $i=0, 1,2,\ldots, n-1$. 

In our application of the construction of deletion-correcting code, the inverse $\rho^{-1}$ is the permutation whose last component is the deleted symbol, and the first $n-1$ components constitute the deleted codeword. Suppose the $\rho^{-1}$ has vector representation
$$
(c_0,c_1,\ldots, c_{n-2}, \Delta).
$$
Then $\pi^{(i)}$ is the inverse of
$$
(c_0, c_1, \ldots, c_{n-i-2}, \Delta , c_{n-i-1},\ldots, c_{n-2}),
$$
which is the vector obtained by inserting the deleted symbol $\Delta$ between symbol $c_{n-j-2}$ and $c_{n-j-1}$.
We write the representation $R(\pi^{(i)})$ as
$$
R(\pi^{(i)}) = \alpha^{(i)}= [a^{(i)}_0, a^{(i)}_1, \ldots, a^{(i)}_{n-1}],
$$
for $i=0,1,\ldots, n-1$.

In addition, we define 
$$\alpha^{(n)}=[a^{(n)}_0,a^{(n)}_1,\ldots, a^{(n)}_{n-2}]$$
to be the representation of $(c_0,c_1,\ldots, c_{n-2})^{-1}$ as a permutation of length $n-1$. This will play an auxiliary role in the derivation.

\begin{example}
	Suppose $\rho=(0,4,1,3,2)$ as in Fig.~\ref{fig:sign}. The vectors $\alpha^{(0)},\ldots, \alpha^{(5)}$ are calculated in Table~\ref{table:vector_a}.
\end{example}

\begin{table}
	\caption{The vectors $\alpha^{(i)}$ for $\rho=(0,4,1,3,2)$. Column sums shown in the last row are distinct. The last column is the auxiliary permutation of length~4}
	\label{table:vector_a}
	\begin{tabular}{|c||c|c|c|c|c||c|}  
		\hline 
		&  $\rho^{-1} =$ & $(\pi^{(1)})^{-1}=$  & $(\pi^{(2)})^{-1}=$ & $(\pi^{(3)})^{-1}=$ & $(\pi^{(4)})^{-1}=$ &  \\
		& $(0,2,4,3,1)$  & $(0,2,4,1,3)$ & $(0,2,1,4,3)$ & $(0,1,2,4,3)$ & $(1,0,2,4,3)$ &  $(0,2,4,3)$\\ \hline
		$j$ & $a^{(0)}_j$ & $a^{(1)}_j$ & $a^{(2)}_j$ & $a^{(3)}_j$ &   $a^{(4)}_j$ & $ a_j^{(5)}$\\ \hline \hline
		0 & 1 & 1 & 1 & 1 & 1 & 1\\
		1 & 1 & 2 & 1 & 1 & 1 & 1 \\
		2 & 2 & 2 & 1 & 3 & 3 & 3\\
		3 & 3 & 3 & 3 & 4 & 1 & 4\\
		4 & 5 & 5 & 5 & 5 & 4 &\\ \hline 
parity & 12 & 13 & 11 & 14 & 10 & \\ \hline
	\end{tabular}
\end{table}

\begin{lemma} \label{lemma:easy}
 With the notation in the previous paragraph, we have the following relations among the variables. 

(i) For $i=2,3,\ldots n-1$,
\begin{equation}
a^{(0)}_i = a^{(1)}_i = \cdots = a^{(i-1)}_i .
\label{eq:left}
\end{equation}

(ii) For $i=0,1,\ldots n-2$,
\begin{equation}
	a^{(i+2)}_i = a^{(i+3)}_i = \cdots = a^{(n-1)}_i = a^{(n)}_i.
	\label{eq:right}
\end{equation}

(iii) For $i=1,2,\ldots, n-1$, $ a_i = a^{(n)}_{i-1} +  b_i$.

(iv) For $i=1,2,\ldots, n-2$,
$$
a^{(i+1)}_{i-1} = a^{(i-1)}_i -  b_i.
$$
\end{lemma}

\begin{proof}
The proof relies on the calculation of of $R(\rho)$ in Theorem~\ref{thm:R_using_inverse}. Suppose the inverse of $\pi^{(i)}$ has vector representation
$$
(c^{(i)}_0, c^{(i)}_1,\ldots, c^{(i)}_{n-1}). 
$$
The value of $a_j^{(i)}$ is the same as the number of indices $k$ in $\{0,1,\ldots, j\}$ that satisfy
\begin{equation}
[c^{(i)}_{n-1-k} - c^{(i)}_{n-1-j}\bmod n] \leq [c^{(i)}_{n-2-j} - c^{(i)}_{n-1-j} \bmod n].
\label{eq:proof1}
\end{equation}

(i). Let $i=2,3,\ldots, n-1$ be a fixed index. The permutations $\pi^{(0)}$, $\pi^{(1)}, \ldots, \pi^{(i-1)}$ have inverses
\begin{align*}
&(c_0, \cdots, c_{n-1-i}, c_{n-i}, \ldots, c_{n-2}, \Delta ) \\
&(c_0, \cdots, c_{n-1-i}, c_{n-i}, \ldots, \Delta, c_{n-2} ) \\
\vdots \\
&(c_0, \cdots, c_{n-1-i}, \Delta, c_{n-i}, \ldots, c_{n-2} ) .
\end{align*}
The value of $ [c^{(i)}_{n-2-j} - c^{(i)}_{n-1-j} \bmod n]$ are the same. Also, in each of the above sequence of numbers, the numbers that are to the right of $c_{n-1-i}$ are the same. There for the number of indices in $\{0,1,\ldots, i+1\}$ that satisfy the condition in \eqref{eq:proof1} are the same.

(ii) The argument is similar to the proof of part (i). Let $i$ be a fixed index between 0 and $n-2$. The inverses of the permutations $\pi^{(i+2)}, \pi^{(i+3)}, \ldots,\pi^{(n-1)}$ are
\begin{align*}
(\Delta, c_0, \ldots, c_{n-i-1}, c_{n-i}, c_{n-i+1}, \ldots, c_{n-2})	& \\
(c_0, \Delta, \ldots, c_{n-i-1}, c_{n-i}, c_{n-i+1}, \ldots, c_{n-2})	& \\
(c_0, \ldots, c_{n-i-1}, \Delta, c_{n-i},  c_{n-i+1}, \ldots, c_{n-2}).	&
\end{align*}
The the subsequences from $c_{n-i}$ to $c_{n-2}$ on the right are all the same. This explains why $a_i^{(i+2)}$ to $a^{(n-1)}_i$ are all identical. Moreover, the value is the same as the component $a_i^{(n)}$ in the representation of $(c_0,c_1,\ldots, c_{n-2})$.

(iii) The task is to compare the representation of
$$
	(c_0,c_1,\ldots, c_{n-2}, \Delta)^{-1},
$$
and
$$
	(c_0,c_1,\ldots, c_{n-2})^{-1}.
$$
We want to show that $a_{i-1}^{(n)}$ and $a_i$ are either equal or differ by 1, and the difference is precisely equal to $b_i$. We can see this by using the procedure before Example~\ref{ex:e} in the calculation of~$a_i$. To determine the component $a_j$ in $R(\rho)$, we cyclically shift the permutation $\rho$ such that the first symbol is $n-1-i$. If the symbol $n-1$ is between the first symbol $n-1-i$ and the symbol $n-2-i$, then $a^{(n)}_{i-1}$ is equal to $a_i+ 1$, and in this case we have $b_i=1$. Otherwise, $a_{i-1}^{(n)}$ is the same as $a_i$.

(iv) It is a consequence of parts (ii) and (iii), plus some book-keeping.
\end{proof}

The value of $a^{(j)}_j$, for $j=0,1,\ldots, n-1$, is of particular importance.  We may call them the pivots.

\begin{lemma} \label{lemma:sum}
 For $i=0,1,\ldots, n-2$,
\begin{equation}
	a^{(i)}_i + a^{(i+1)}_i = i+2.
	\label{eq:sum}
\end{equation} 
\end{lemma}

\begin{proof}  
We observe that the symbols at locations $n-2-i$ and $n-1-j$ in $(\pi^{(i)})^{-1}$ and  $(\pi^{(i+1)})^{-1}$ are the same, but in different order,
\begin{align*}
	(\pi^{(i)})^{-1} &= (\ldots, c_{n-2-i}, \Delta , \underbrace{c_{n-1-i}, \ldots, c_{n-2}}_{i}), \\
	(\pi^{(i+1)})^{-1} &= (\ldots, \Delta, c_{n-2-i}, \underbrace{c_{n-1-i}, \ldots, c_{n-2}}_{i}). 
\end{align*}
In both sequences, the $i-1$ numbers on the right part are the same. By \eqref{eq:alpha}, $a^{(i)}_i$ is equal to one plus the number of indices $k$ between 0 and $i-1$ that satisfy
$$
[c^{(i)}_{n-1-k} -  \Delta \bmod n] \leq [c_{n-2-i} -\Delta \bmod n].
$$

On the other hand, by \eqref{eq:alpha2},  $a^{(i)}_i$ is equal to one plus the number of indices $k$ between 0 and $i-1$ that satisfy
$$
[c^{(i)}_{n-1-k} -\Delta \bmod n] \leq [\Delta - c_{n-2-i}  \bmod n].
$$
Each $k\in\{0,1,\ldots, i-1\}$ must satisfy exactly one of the above two inequalities. Therefore, the sum $a^{(i)}+a^{(i+1)}$ is equal to $i+2$.
\end{proof}

The next lemma illustrates how to recursively compute the sequence $a^{(i)}_i$, for $i\geq 0$.
	
\begin{lemma}	\label{lemma:key}
For $i=0,1,\ldots, n-2$, we have
$$
a_{i+1}^{(i+1)} = a^{(i)}_i + a_{i+1} - (i+2)b_{i+1}.
$$
\end{lemma}

\begin{proof} The main proof idea is to find a relationship between $a^{(i)}_i$, $a^{(i+1)}_{i+1}$, and $a^{(i+2)}_{i}$. The inverse of $\pi^{(i)}$, $\pi^{(i+1)}$, and $\pi^{(i+2)}$ are, respectively
\begin{align*}
(\pi^{(i)})^{-1} =(c_0, \ldots, c_{n-4-i}, y,x,\Delta, \underbrace{c_{n-1-i}, \ldots, c_{n-2}}_i) ,& \\
(\pi^{(i+1)})^{-1} = (c_0, \ldots, c_{n-4-i}, y,\Delta,x, \underbrace{c_{n-1-i}, \ldots, c_{n-2}}_i), & \\
(\pi^{(i+2)})^{-1} =(c_0, \ldots,c_{n-4-i},  \Delta, y,x,\underbrace{c_{n-1-i}, \ldots, c_{n-2}}_i). & 
\end{align*}	
To simplify notation, we will write $x$ to stand for $c_{n-2-i}$ and $y$ for $c_{n-3-i}$. We let $\mathcal{S}$ denote the set
$$
\mathcal{S}_i \triangleq \{ c_{n-1-i}, c_{n-i}, \ldots, c_{n-2}\},
$$
consisting of the $i$ entries on the right of the sequences.

By \eqref{eq:alpha}, we know that $a^{(i)}_i$ is equal to one plus the number of integers $z$ in $\mathcal{S}_i$ that satisfy
\begin{equation}
[z - \Delta \bmod n] \leq [x - \Delta \bmod n].
\label{eq:solution1}
\end{equation}

Express  $a^{(i+2)}_i$ as one plus the number of integers $z$ in $\mathcal{S}_i$ that satisfy
\begin{equation}
[z - x \bmod n] \leq [y - x \bmod n].
\label{eq:solution2}
\end{equation}
The value of $a^{(i+1)}_{i+1}$ can be written as one plus the number of integers $w$ in $\mathcal{S}_i \cup \{x\}$ that satisfy
\begin{equation}
[w - \Delta  \bmod n] \leq [y - \Delta \bmod n].
\label{eq:combined}
\end{equation}

We consider two cases below.

In the first case, we assume that $b_{i+1}=0$. In this case, the difference $y-\Delta$ is less than $n$, and we have
$$
[x - \Delta \bmod n]+ [y - x \bmod n] =
[y - \Delta \bmod n].
$$
Because $x$ satisfies the condition in \eqref{eq:combined}, we have
$$	a^{(i+1)}_{i+1} = a^{(i)}_i + a^{(i+2)}_i,
$$
and by part (iii) in Lemma~\ref{lemma:easy}, we obtain
\begin{equation}
	a^{(i+1)}_{i} = a^{(i)}_i + a_{i+1}.
	\label{eq:part1}
\end{equation}

In the second case, we assume that $b_{i+1}=1$. This means that
$$
[x - \Delta \bmod n]+ [y - x \bmod n] =
[y - \Delta \bmod n] + n.
$$
The total number of solutions to \eqref{eq:solution1} and \eqref{eq:solution2} is the number of solution to \eqref{eq:combined} plus~$i$. Moreover, because $b_{i+1}=1$, $x$ does not satisfy \eqref{eq:combined}. We thus obtain
$$
a^{(i+1)}_{i+1} = a^{(i)}_i + a^{(i+2)}_i - (2+j) + 1
$$
We now apply part (iii) of Lemma~\ref{lemma:easy} to get
\begin{equation}
	a^{(i+1)}_{i+1} = a^{(i)}_i + a_{i+1} - (2+j) 
	\label{eq:part2}
\end{equation}

Finally, we combine \eqref{eq:part1} and \eqref{eq:part2} into a single equation,
$$	a^{(i+1)}_{i+1} = a^{(i)}_i + a_{i+1} - (2+i)b_{i+1} 
$$
This complete the proof of Lemma~\ref{lemma:key}.
\end{proof}

Using Lemma~\ref{lemma:key}, we can compute the values of $a_i^{(i)}$, for $i=0,1,\ldots, n-1$, by
\begin{equation}
	a^{(i)}_i = 1 + \sum_{j=1}^{i} a_j - \sum_{j=1}^i (j+1)b_j.
	\label{eq:pivot}
\end{equation}

We now prove Theorem~\ref{thm:distint_values} by mathematical induction. We first consider the case $i=1$ by comparing the parity of $R(\pi^{(0)}) = [a_0,a_1,\ldots, a_{n-1}]$ and  $R(\pi^{(1)}) = [a^{(1)}_0, a^{(1)}_1, \ldots,  a^{(1)}_{n-1}]$. The two vectors  $R(\pi^{(0)})$ and  $R(\pi^{(1)})$ only differ in the second component. Putting $i=1$ in \eqref{eq:pivot}, we get 
$$
a_1^{(1)} = 1 + a_1^{(0)} - 2b_1.
$$
Therefore,
$$
\textit{parity}(R(\pi^{(1)}))
-\textit{parity}(R(\pi^{(0)})) =
a_1^{(1)} - a_1^{(0)} =
1 -  2b_1  = (1-b_1) - b_1,
$$
which is the same as \eqref{eq:distinct_values} with $i$ substituted by~1.

We next consider \eqref{eq:distinct_values} for $i>1$. From parts (i) and (ii) in Lemma~\ref{lemma:easy}, the two vectors $\textit{parity}(R(\pi^{(i)}))$ and
$\textit{parity}(R(\pi^{(i+1)}))$ differ in position $i-1$, $i$ and $i+1$. We need to compute six values  $a^{(i)}_{i-1}$, $a^{(i)}_{i}$, $a^{(i)}_{i+1}$,
$a^{(i+1)}_{i-1}$, $a^{(i+1)}_{i}$, and $a^{(i+1)}_{i+1}$.

By Lemma~\ref{lemma:sum} and part (iv) of Lemma~\ref{lemma:easy}
%
%
\begin{align*}
& \phantom{=} \textit{parity}(R(\pi^{(i+1)})) - \textit{parity}(R(\pi^{(i)})) \\
&= (a^{(i+1)}_{i+1} + a^{(i+1)}_{i}) + a^{(i+1)}_{i-1} 
- a^{(i)}_{i+1} - a^{(i)}_{i} - a^{(i)}_{i-1} \\
& =  (a^{(i+1)}_{i+1} - a^{(i)}_{i} ) - a^{(i)}_{i+1} + a^{(i+1)}_{i-1} +  a^{(i+1)}_{i}  - a^{(i)}_{i-1}  \\
&=  a^{(i+1)}_{i+1} - a^{(i)}_{i}   - a^{(i)}_{i+1} +(a^{(i-1)}_{i} - b_i) +(i+2-a^{(i)}_i) - (i+1-a^{(i-1)}_{i-1}) \\
&=  (a^{(i+1)}_{i+1} - a^{(i)}_{i})
-  (a^{(i)}_{i} - a^{(i-1)}_{i-1})
+(a^{(i-1)}_{i} - a^{(i)}_{i+1})
 + 1 -b_i .
\end{align*}

By applying Lemma~\ref{lemma:easy} and Lemma~\ref{lemma:key}, we can simplify it to
\begin{align*}
&a_{i+1} - (i+2)b_{i+1} - (a_i - (i+1)b_i) + a_i - a_{i+1} + 1-b_i \\
& = 1 + ib_i - (i+2)b_{i+1}.
\end{align*}

Apply the induction hypothesis,
\begin{align*}
	\textit{parity}(R(\pi^{(i+1)})) &=
	\textit{parity}(R(\pi^{(i)})) +
(	\textit{parity}(R(\pi^{(i+1)}))  - 	\textit{parity}(R(\pi^{(i)})) )	\\
& =	\textit{parity}(R(\pi^{(0)})) + i(1-b_i) - \sum_{j=1}^i b_j + (1 + ib_i - (i+2)b_{i+1} )\\
&= i + 1 - i b_{i+1} - b_{i+1} - \sum_{j=1}^{i+1} b_j \\
&= (1+i)(1-b_{i+1}) - \sum_{j=1}^{i+1} b_j.
\end{align*} 
This completes the proof of Theorem~\ref{thm:distint_values}.

\section{Encoding and decoding algorithm}
\label{sec:algorithm}

A key in the encoding and decoding algorithms is the computation of the function $R$ and its inverse. A naive implementation would require $O(n^2)$ steps. With the use of more advanced data structure, we can reduce the time complexity significantly.

We first consider the computation of $R$. Let $\rho$ be a given permutation $S_n$. We will obtain the components in $R(\rho) = [a_0,a_1,\ldots, a_{n-1}]$ in reverse order, starting from $a_{n-1}$. Let $\iota$ denote the inverse of $\rho$, which can be computed in $O(n)$ time. The last component $ a_{n-1}$ in $R(\rho)$ is equal to $n - \iota(0)$.

\begin{algorithm}
	\caption{Calculation of $R(\rho)$}
	\label{alg:red_black}
	{\bf Input:} Integer $n$, a permutation $\rho$ in $S_n$. \\
	{\bf Output:} A vector in $V_n$.
	\begin{algorithmic}[1]
		\STATE Compute $\iota \gets \rho^{-1}$
		\STATE Define $\iota(-1)$ as $n$
		\STATE Initialize an empty red-black tree
		\FOR{$j = 0,1,\ldots, n-1$}
		\IF{$\iota(j-1) > \iota(j)$}
		\STATE $x \gets$ no. of symbols in $\{0,1,\ldots, \iota(j-1)-1\}$ that are in the red-black tree
		\STATE $y \gets$ no. of symbols in $\{0,1\ldots, \iota(j)-1\}$ that are in the red-black tree
		\STATE Let $a_{n-1-j} \gets (\iota(j-1) - \iota(j) ) - (x-y)$
		\ELSE
		\STATE $x \gets$ no. of symbols in $\{\rho(0), \rho(1) , \ldots, \rho(j-1)-1\}$ that are  in the red-black tree
		\STATE $y \gets$ no. of symbols in $\{\rho(j), \rho(j+1)\ldots, \rho(n-1)\}$ that are in the red-black tree
		\STATE Let $a_{n-1-j} \gets [\iota(j-1) - \iota(j) \bmod n] - (x+y)$
		\ENDIF
		\STATE Insert $\iota(j)$ to the red-black tree
		\ENDFOR
		\RETURN $(a_0,a_1,\ldots, a_{n-1})$
	\end{algorithmic}
\end{algorithm}

For other values of $j$, we recall that we can compute $a_j$ by first cyclically rotating $\rho$ such that $n-1-j$ is located at the first position on the left. The value of $a_j$ is then the number of symbol to the left of symbol $n-2-j$ that is larger than $n-2-j$. For ease of computation, we need not do the cyclic rotation, but we can use the information in the inverse permutation $\iota$. The total number of symbols between $n-1-j$ and $n-2-j$ is equal to
$$
[\iota_{n-1-j} - \iota_{n-2-j} \bmod n].
$$
Hence, the remaining task is to count the number of symbols that are between symbol $n-1-j$ and symbol $n-2-j$. This can be done by using the data structure called red-black tree, which is a binary tree that are close to balanced. Each node has a label. The label of a node is larger than the label of the left child, and is smaller than the label of the right child. We can build a red-black tree by inserting numbers one by one, and the insertion operation takes $O(\log n)$ time. 
Detail of the implementation of red-black tree can be found in~\cite[Chapter 13]{IntroductionToAlgorithms}. Using this data structure, we can calculate the representation $R(\rho)$ using Algorithm~\ref{alg:red_black}. 

We note that the operations in step 6, 7, 10, 11 and 14 can be performed in $O(\log n)$ time. The time complexity of the whole algorithm is $O(n\log n)$.

In order to compute the inverse function $R^{-1}(\rho)$ efficiently, we introduce another data structure from~\cite{Bille2017}, called ``fast dynamic array''.  This is a data structure that implements a 1-dimensional array. We can insert a new cell in the array in $O(n^\epsilon)$ time, where $\epsilon$ is a fixed positive parameter, while the access time is $O(1)$.

To circumvent the difficulty of cyclic shift, we will replicate a permutation and double the number of symbols. This will not affect the order of time complexity. We illustrate the computation using an example. Suppose $n=5$ and $\alpha$ is the vector 
$$(a_0,a_1,a_2,a_3,a_4) = (1,2,1,4,3).$$ We start from the initial string of symbols
$$
3434
$$ 
and we will add the other symbols one by one. The symbol 2 is inserted after 4, because $a_1$ is equal to 2,
$$
342342.
$$
Since $a_2=1$, the symbol 1 is inserted after symbol 2,
$$
34213421.
$$
Symbol 0 is put after the symbol 2 as $a_3=4$, and we obtain
$$
3420134201.
$$
Finally, because $a_5=2$, symbol 0 should be in the middle of the permutation. Hence, the permutation that corresponds to $\alpha = (1,2,1,4,3)$ is
$$
(4,2,0,1,3).
$$

We implement this idea in Algorithm~\ref{alg:inverseR}, with the fast insertion as in~\cite{Bille2017}. As in the rest of the paper, the index of array $A$ starts from~0. In the algorithm below, the variable $x$ is location of the newly inserted symbol.

\begin{algorithm}
	\caption{Computation of  $R^{-1}(\alpha)$}
	\label{alg:inverseR}
	{\bf Input:} Integer $n$, a vector $\mathbf{\alpha} = (a_0,a_1,\ldots, a_{n-1})$ in $V_n$ \\
	{\bf Output:} A permutation in $S_n$.
	\begin{algorithmic}[1]
\STATE  Initialize array $A$ to $[n-2, n-1, n-2, n-1]$
\STATE $L=2$ \qquad \qquad // $L$ is the length of array $A$ divided by 2
\STATE $x \gets 1$ \qquad \qquad // Initially, $x$ is a pointer to symbol $n-1$
\FOR{$j= n-3, n-2, \ldots, 0$}
\STATE $x \gets x+a_{n-2-j}$
\STATE $L \gets L+1$
\STATE Insert symbol $j$ at position $x$ 
\IF{$x \geq L$}
  \STATE $x \gets x-L$
  \STATE Insert symbol $j$ at position $x$
\ELSE
  \STATE Insert symbol $j$ at position $x+L$
\ENDIF
\ENDFOR	
\STATE $i \gets x - a_{n-1}+1$
\IF {$x - a_{n-1} +1 < 0 $}
\STATE $i \gets i+n$
\ENDIF 
\RETURN $(A[i], A[i+1], \ldots, A[i+n-1])$
	\end{algorithmic}
\end{algorithm}

In Algorithm~\ref{alg:inverseR},  the insertion operation is performed $O(n)$ times. If each insertion takes $O(n^\epsilon)$ time, the time complexity of Algorithm~\ref{alg:inverseR} is $O(n^{1+\epsilon})$, for some parameter $\epsilon>0$.

Using the inverse function $R^{-1}$, we can perform encoding as in Algorithm \ref{alg:encode}. The data to be encoded is drawn from the Cartesian product 
$$
\{1,2\} \times \{1,2,3\} \times \cdots \times \{1,2,\ldots, n-1\}.
$$
The most time-consuming step is the computation of $R^{-1}(\alpha)$ in step~3, which takes $O(n^{1+\epsilon})$ time. Since $\epsilon$ can be any positive real number, the overall complexity is close to linear time.

\begin{algorithm}
	\caption{Encoding algorithm}
	\label{alg:encode}
	{\bf Input:} Integer $n$, $T$ \\
	\phantom{Input:\ }
	  $a_j\in\{1,2,\ldots, j+1\}$ for $j=1,2,\ldots, n-2$ \\
	{\bf Output:} A permutation in $S_n$.
	\begin{algorithmic}[1]
		\STATE  Set $a_{n-1} \gets \langle T - \sum_{j=1}^{n-2} a_j \rangle_n$ \qquad // Set the value of $a_{n-1}$ such that the parity is $T$
		\STATE Let $\alpha$ be the vector $[1,a_1,a_2,\ldots, a_{n-1}]$
		\STATE Compute  $\rho \gets R^{-1}(\alpha)$
		\RETURN the permutation $\rho^{-1}$  \qquad // The inverse of $R^{-1}(\alpha)$ is the codeword
	\end{algorithmic}
\end{algorithm}

Decoding algorithm is in Algorithm~\ref{alg:decode}. We first determine the missing symbol, denoted by $\Delta$ in the algorithm, and append it at the end of vector $\mathbf{v}$. This will form a permutation $\rho$. If the parity of $\rho$ is equal to $T$, then we will skip steps 7 to 20, because $\rho$ already gives the correct answer. Otherwise, if the parity of $\rho$ is not $T$, then we step-wise decrease the value of $\rho_\Delta$ until the parity is equal to $T$. The parity of $\rho$ within the for-loop is computed via the sequence $b_1, b_2,\ldots, b_{n-1}$, which is defined as in~\eqref{eq:b}. In Step 16, we quit the for-loop if the parity is equal to~$T$.

We note that the representation function $R$, which is the bottleneck of the time complexity, is computed two times. The total time complexity of Algorithm~\ref{alg:decode} is $O(n \log n)$.

\begin{algorithm}
	\caption{Decoding algorithm for one deletion}
	\label{alg:decode}
	{\bf Input:}  Integers $n$, $T$ \\
	\phantom{Input:} Vector $\mathbf{v} = (v_0, v_1, v_2,\ldots, v_{n-2})$, with distinct components in $\{0,1,\ldots, n-1\}$. \\
	{\bf Output:} $a_j\in\{1,\ldots, j+1\}$,  for $j=1,2,\ldots, n-2$.
	\begin{algorithmic}[1]
\STATE Let $\Delta$ be the integer in $\{0,1,\ldots, n-1\}$ that is not a component in $\mathbf{v}$
\STATE Let $\pi$ be the permutation $(v_0, v_1,\ldots, v_{n-2}, \Delta)$
\qquad // $\Delta$ is the deleted symbol
\STATE Set $\rho \gets \pi^{-1}$ \quad // Compute the inverse of $\pi$
\STATE Calculate  $p \gets \textit{parity}(R(\rho))$ 
\STATE $m_1 \gets p$, $m_2\gets p$ \quad // Initialize $m_1$ and $m_2$ to the parity $p$ of $R(\rho)$
\STATE  Compute $b_j$, for $j=1,2,\ldots, n-1$ by the formula in  \eqref{eq:b}
\IF {$p \not\equiv T\bmod n$}
  \FOR{$k= 1, 2, ,3 \ldots, n-1$}
  \STATE $i \gets \pi_{n-1-k}$
  \STATE Exchange the value of $\rho_i$ and $\rho_\Delta$
  \IF {$b_k = 1$} 
    \STATE $m_1 \gets m_1-1$ \qquad // Using  Theorem~\ref{thm:distint_values}, update the parity after the exchange 
  \ELSE
    \STATE $m_2 \gets m_2+1$
  \ENDIF
  \IF {$m_1 \equiv  T \bmod n$ or $m_2 \equiv T \bmod n$}
    \STATE Break   \quad // proceed to step 21 if we have the correct parity
  \ENDIF
  \ENDFOR
\ENDIF
\STATE Compute $\alpha \gets R(\rho)$ \quad // Compute the representation of permutation $\rho$
\STATE Let $a_j \gets $ component $j$ in $\alpha$, for $j=1,2,\ldots, n-2$
\RETURN $(a_1,a_2,\ldots, a_{ n-2})$ \quad // Components $a_1$ to $a_{n-2}$ in $\alpha$ are the data
	\end{algorithmic}
\end{algorithm}

\section{Concluding Remarks}

We consider a special representation $R(\pi)$ that is based on block transposition. Using this representation, we can give an  alternate description of Levenshtein's perfect permutation code, and devise quasi-linear time encoding and decoding algorithm that can correct one deletion. In fact, the computation  of the function $R(\pi)$ is amenable to parallel processing. If parallel computing is available, we can further reduce the time complexity to $O(n)$.


\end{document}